%% file: mdtfPaper.tex
\documentclass[letterpaper]{article}
\pdfoutput=1
\usepackage{aaai}
\usepackage{times}
\usepackage{helvet}
\usepackage{courier}
\usepackage{amsmath}
\usepackage{amsfonts}
\usepackage{amsthm}
\usepackage{graphicx}
\usepackage{color}
\usepackage{algorithm}
\usepackage{algorithmic}
\newtheorem{observation}{Observation}
\newtheorem{proposition}{Proposition}
\newtheorem{definition}{Definition}

\newcommand{\leaveout}[1]{}

\allowdisplaybreaks

\begin{document}
\title{Mechanism Design for Team Formation}
\author{Mason Wright\\Computer Science \& Engineering\\University of
  Michigan\\Ann Arbor, MI\\masondw@umich.edu \And Yevgeniy
  Vorobeychik\\Electrical Engineering and Computer Science\\Vanderbilt
University\\Nashville, TN\\yevgeniy.vorobeychik@vanderbilt.edu}

\maketitle

\begin{abstract}
\begin{quote}
Team formation is a core problem in AI. Remarkably, little prior work has addressed the problem of mechanism design for team formation, accounting for the need to elicit agents' preferences over potential teammates. Coalition formation in the related hedonic games has received much attention, but only from the perspective of coalition stability, with little emphasis on the mechanism design objectives of true preference elicitation, social welfare, and equity.
We present the first formal mechanism design framework for team formation, building on recent combinatorial matching market design literature. We exhibit four mechanisms for this problem, two novel, two simple extensions of known mechanisms from other domains. Two of these (one new, one known) have desirable theoretical properties. However, we use extensive experiments to show our second novel mechanism, despite having no theoretical guarantees, empirically achieves good incentive compatibility, welfare, and fairness.
\end{quote}
\end{abstract}

\input{intro}
\input{mdproblem}

\input{mechanisms}

\input{experiments}
\input{conclusion}





\bibliographystyle{aaai}
\small
\bibliography{mdtfPaper}

\input{appendix}

\end{document}

%% file: intro.tex
\section{Introduction}

Teamwork has been an important and often-studied area of artificial intelligence research. 
Typically, the focus is on coordinating agents to achieve a common goal.
The complementary problem of team formation considers how to form high-quality teams, whose agents have skills that are jointly well suited for a task~\cite{Marcolino13}.
Notable team formation applications include formation of research teams, class project groups, groups of roommates, or disaster relief teams.

Many prior team formation studies have assumed that agents are indifferent about which other agents they are teamed with, or have preferences known to the team formation mechanism.
Models dealing with known agent preferences over teammates, termed \emph{hedonic games}, have seen an extensive literature since being introduced by~\citeauthor{aumann1974cooperative}~\shortcite{aumann1974cooperative}.
In a hedonic game, the mechanism is given a set of agents, each having public preferences over which others might be on its team; the mechanism must partition the agents into teams based on their preferences.

Past research on hedonic games has focused on the problem of forming \emph{stable} coalitions, from which no set of agents would prefer to defect. 
Since a core partition may not exist in a hedonic game, even when preferences of players are additively separable~\cite{banerjee2001core}, 
much research is focused on alternative notions of stability, or on highly restricted agent preferences~\cite{bogomolnaia2002stability,alcalde2004researching,cechlarova2001stability}, or on the time complexity of testing core emptiness~\cite{ballester2004np,sung2010computational}.

We consider team formation as a mechanism design problem, where individuals have preferences over teammates, as in hedonic games.
As in traditional mechanism design (and unlike hedonic games), we assume that these preferences are private and must be elicited in order to partition players reasonably into teams.
We draw a connection to another budding literature, that of combinatorial matching market design, which has course allocation as a typical application~\cite{budish2010multi}.

An important concern in combinatorial matching, which we inherit, is the ex post fairness of allocations. For example, consider a simple randomized mechanism, \emph{random serial dictatorship}, which has been proposed for course allocation and is readily adapted to team formation. In random serial dictatorship, agents are randomly ordered by the mechanism and then take turns, in order, selecting their entire teams from among the remaining agents. Random serial dictatorship is \emph{strategyproof}, meaning that it is a dominant strategy for any agent to report its true preferences over teams. Random serial dictatorship is also \emph{ex post Pareto efficient}, in that any allocation it returns cannot be modified to improve an agent's welfare without reducing some other agent's (assuming no indifferences). But this mechanism results in a highly inequitable distribution of outcomes ex post.

Budish and Cantillon~\shortcite{budish2010multi} proposed a more sophisticated alternative, \emph{approximate competitive equilibrium from equal incomes (A-CEEI)}, which is \emph{strategyproof-in-the-large} (i.e., when the number of players becomes infinite), and provably approximately fair~\cite{budish2010multi,budish2011combinatorial}.
The work on combinatorial matching in turn follows earlier work on bipartite matching and school choice~\cite{roth1999redesign,abdulkadiroglu2003school}.


Our contributions are as follows. 
\begin{enumerate}
\item We present the problem of mechanism design for team formation, focused on achieving (near-)incentive compatible preference reporting, high social welfare, and fair allocation.  This problem is closely related to both combinatorial and bipartite matching market design, but is distinct from both in two senses: first, the matching is not bipartite (players match to other players), and therefore typical matching algorithms which only guarantee strategyproofness for one side are unsatisfactory; and second, mechanisms used in combinatorial exchanges to provide fairness guarantees are not directly applicable, as they rely on having a fixed set of items which are the subject of the match and which are not themselves strategic;
\item we extend two well-known mechanisms (random serial dictatorship and Harvard Business School draft) used for combinatorial matching to our setting;
\item we propose two novel mechanisms for our setting (A-CEEI for team formation, or A-CEEI-TF, and one-player-one-pick draft, or OPOP); 
\item we prove that A-CEEI-TF is approximately fair and strategyproof-in-the-large;
\item we offer empirical analysis of all mechanisms, which shows that our second mechanism, OPOP, outperforms others on most metrics, and has better incentive properties than A-CEEI-TF.
\end{enumerate}
An important and surprising finding of our investigation is that the simple draft mechanism we propose empirically outperforms the more complex A-CEEI-TF alternative by a large margin in fairness and incentive compatibility, even while A-CEEI-TF has more compelling theoretical guarantees.

%% file: mdproblem.tex
\section{Mechanism Design Problem}

Our point of departure is the formalism of \emph{hedonic games}.
We define a \emph{hedonic game} as a tuple $(N, \succ)$, where $N$ is the set of players, and $\succ$ is a vector containing each player's preference order over sets of other players that it could be teamed with. 
The task is to partition the players in $N$ into a coalition structure, where each player is in exactly one coalition.

We assume that player preferences are \emph{additively separable}~\cite{aziz2011stable}, which means that there exists an assignment of values $u_i(j)$ for all players $i$ and their potential teammates $j$, so that $i$'s total utility of a subset of others $S$ is $\sum_{j \in S} u_i(j)$ (which induces a corresponding preference ordering over subsets of possible teammates).
In addition, we assume that $u_i(j) \ge 0$ for all $i,j$. 

These assumptions are useful for two reasons. First, in many data sets that record preferences of individuals over others, the preferences are entered as non-negative values for individuals, as in rank order lists or Likert ratings of individuals. Additive separable preferences are the most natural way to induce preferences over groups from such data. Second, many prior studies in team formation and the related domain of course allocation have assumed that agents have non-negative, additive separable preferences, as in the $\mathcal{B}$-preferences of \cite{cechlarova2001stability} and in the bidding points auction.

Most prior work on hedonic games focuses on coalition stability.
Our goal is distinct: We take as input player preferences over teams (that is, over others that they could be teamed with), which we assume to be additive with non-negative values, and output a partition of the players into teams.
We assume that it is subsequently difficult for players to alter team membership.
Our primary challenge, therefore, is to encourage players to report their preferences honestly, and form teams that are fair and yield good teammate matchings; all three notions shall be made precise presently.
Note that in this construction we assume that no money can change hands (unlike the work by \citeauthor{Li04}~\shortcite{Li04}).

Observe that in our model, all players always prefer to be put on a single team (since values for all potential teammates are positive).
In reality, many team formation problems have hard constraints on team sizes (or, equivalently, on the number of teams), particularly when multiple tasks need to be accomplished.
For example, project teams usually have an upper bound on size.
We capture this by introducing team size constraints; formally, the size of any team must be in the interval $[\underline{k},\overline{k}]$, with $\underline{k} \ge 1$, $\overline{k} \le |N|$, and $\underline{k} \le \overline{k}$.
For example, if a classroom with $25$ students must be divided into $6$ approximately equal-size teams, we could have $\underline{k} = 4$ and $\overline{k} = 5$. 
We assume throughout that the specific values of $\underline{k}$ and $\overline{k}$ admit a feasible allocation. (This is not always the case; see supplemental material for details.)

In contrast with a typical approach in mechanism design, which seeks to maximize a single objective such as social welfare or designer revenue, subject to a constraint set, we take an approach from the matching market design literature, and seek a collection of \emph{desirable properties} (see, e.g., \citeauthor{Budish12}~\shortcite{Budish12}).
Specifically, we consider three properties: incentive compatibility, social welfare, and fairness.
Given the fact that all three cannot be achieved simultaneously in our setting, we will analyze the extent to which each can be achieved through specific mechanisms.

\paragraph{Incentive Compatibility}
Incentive compatibility holds if there is no incentive for an agent to misreport its preferences.
We consider two forms of incentive compatibility: \emph{strategyproofness}, which means that it is a dominant strategy for any agent to report its true preferences, and \emph{ex post equilibrium}, which means that it is a Nash equilibrium for all agents to report their true preferences.
The former will be considered in theoretical analysis, while the latter will be the focus of empirical incentive assessment.
In particular, our theory will focus on \emph{strategyproofness-in-the-large}~\cite{budish2011combinatorial}, defined as follows.
Consider a market where each agent has been replaced with a measure-one continuum of replicas of itself, such that each individual agent has zero measure and all agents are price takers. A mechanism is strategyproof-in-the-large if, in such a market, it is a dominant strategy for each agent to reveal its true preferences. An example of a mechanism that is not strategyproof-in-the-large is the Harvard Business School draft considered below, in which an agent may benefit from misreporting its preferences, regardless of its own measure relative to the market size \cite{budish2010multi}.
In empirical analysis, in contrast, we determine a lower bound on the \emph{regret of truthful reporting}, which is the most any agent can gain ex post by misreporting preferences when all others are truthful.




\paragraph{Social Welfare}
As in traditional mechanism design, we consider social welfare as one of our primary design criteria.
Social welfare is just the sum of player utilities achieved by a specific partition of players into teams.
Formally, if $\mathcal{Q}$ is a partition of players, social welfare is defined as $SW(\mathcal{Q}) = \frac{1}{|N|}\sum_{S \in \mathcal{Q}} \sum_{i,j \in S} u_i(j)$.
In addition, we consider the weaker notion of \emph{ex post Pareto optimality} when discussing alternative mechanisms and their theoretical properties.
A partition of players $\mathcal{Q}$ is ex post Pareto optimal if no other partition strictly improves some agent's utility without lowering the utility of any other agent.

\paragraph{Fairness}
The measure of fairness we consider is \emph{envy-freeness}.
An allocation is \emph{envy-free} if each agent weakly prefers its own allocation to that of any other agent. 
An approximate notion of envy-freeness that we adopt from Budish~\shortcite{budish2011combinatorial} is \emph{envy bounded by a single teammate}, in which any allocation an agent prefers to their own ceases to be preferred through removal of a single teammate from it.\footnote{In the supplemental material we discuss another measure of fairness.}
The following negative result makes apparent the considerable challenge associated with the design problem we pose.


\begin{proposition}
\label{T:envyBoundedBySingle}
There may not exist a partition of players that bounds envy by a single teammate.
\end{proposition}

\subsubsection{Proof} Consider a team formation problem with $6$ agents, $\{A, B, C, D, E, F\}$, $\underline{k} = 3$, $\overline{k} = 3$, so that two equal-size teams must be formed. The agents' additive separable preferences are encoded in Table \ref{tab:envyBoundSingle}.

\begin{table}[h!]
\centering
\begin{tabular}{ c | c c c c c c c }
     & $A$ & $B$ & $C$ & $D$ & $E$ & $F$ \\
  \hline
  $A$ & x & 0 & 1 & 2 & 4 & 8 \\
  $B$ & 8 & x & 4 & 2 & 1 & 0 \\
  $C$ & 8 & 0 & x & 4 & 2 & 1 \\
  $D$ & 8 & 1 & 0 & x & 4 & 2 \\
  $E$ & 8 & 2 & 1  & 0 & x & 4 \\
  $F$ & 8 & 4 & 2 & 1 & 0 & x \\
\end{tabular}
\caption{Each row $i$ encodes the additive separable value for agent $i$ of each other agent.}
\label{tab:envyBoundSingle}
\end{table}

No partition of these agents into two teams of size $3$ gives every agent envy bounded by a single teammate. To see this, consider that each agent other than $A$ has a bliss point on a team with $A$ and one other agent, where the second agent is $C$ for agent $B$, $D$ for agent $C$, and so on until ``wrapping around'' with $B$ for agent $F$. Three of the agents will not be on a team with agent $A$, and at least one of these agents, say agent $i$, will not be on a team with its second-favorite agent either. Some other agent $j$ must then be on a team with the two most-preferred agents of the player $i$. By construction, player $i$ is on a team of value $3$ or less, while the team of agent $j$ has value $12$ to agent $i$, and value $4$ to agent $i$ with its more valuable player (player $A$) removed. Therefore, envy cannot be bounded by a single teammate for all agents. \qed

%% file: mechanisms.tex
\section{Team Formation Mechanisms}

We describe four mechanisms for team formation:
two are straightforward applications of known mechanisms, while two are novel.

\subsection{Random Serial Dictatorship}

\emph{Random serial dictatorship} (RSD) has previously been proposed in association with school choice problems~\cite{abdulkadiroglu2003school}.
In RSD, players are randomly ordered, and each player chosen in this order selects his team (with players thereby chosen dropping out from the order).
The process is repeated until all players are teamed up.

\begin{proposition}
\label{T:rsdic}
Random serial dictatorship is strategyproof, and ex post Pareto efficient as long as players choosing later cannot choose a larger team. \footnote{The proofs of this and other results are in the supplemental material.}
\end{proposition}
While RSD is ex post Pareto efficient, this turns out to be a weak guarantee, and does not in general imply social welfare maximization, something that becomes immediately apparent in the experiments below.
Envy-freeness is, of course, out of the question due to Proposition~\ref{T:envyBoundedBySingle}.

\subsection{Harvard Business School (HBS) Draft}

Players are randomly ordered, with the first $T$ assigned as captains.
We then iterate over captains, first in the random order, then in reverse, alternating. The current team captain selects its most-preferred remaining player to join its team, based on its reported preferences. 
\begin{proposition}
\label{T:hbs}
HBS draft is not strategyproof or ex post Pareto efficient.
\end{proposition}

\subsection{One-Player-One-Pick (OPOP) Draft}

Players are randomly ordered. Given the list of team sizes, the first $T$ players are assigned to be captains of the respective teams. Then iterate over the complete player list. If the next player is a team captain, it selects its favorite unassigned agent to join its team. If the next player is unassigned, it will be assigned to join its favorite incomplete team (as defined below), and if the team still has space, this player chooses its favorite unassigned agent to join them.
We define a ``favorite'' incomplete team for an agent as follows.
Let $S$ be an incomplete team with $v_S$ vacancies. Let the mean value to player $i$ of the unassigned players be $\mu_i$. 
We then assign the following utility of an incomplete team $S$ to agent $i$:
\(\sum_{j \in S} u_i(j) + (v_S - 1) \mu_i.\)

\begin{proposition}
\label{T:opopic}
The One-Player-One-Pick draft is not strategyproof or ex post Pareto efficient.
\end{proposition}

\leaveout{
 It may be less obvious that the OPOP draft is not strategyproof, compared to the HBS draft, because in the OPOP draft, each player makes at most one selection. In the HBS draft, a player can benefit from under-reporting its value for unpopular items, because such items are likely to remain available after multiple rounds of the draft, while less-valuable but more popular items are likely to be taken in an early round. In the OPOP draft, no such incentive to lie exists, because each player makes at most one selection.

Note that the HBS draft and OPOP draft are strategyproof for hedonic games with $\overline{k} = 2$, because in such cases these mechanisms are equivalent to random serial dictatorship, as each team captain chooses only one other player to join it.
}


\subsection{Competitive Equilibrium from Equal Incomes}

We now propose a more complex mechanism, based on the \emph{Competitive Equilibrium from Equal Incomes (CEEI)}, which is explicitly designed to achieve allocations that are more ex post fair than the alternative mechanisms.


We begin by defining CEEI, previously introduced by Varian~\shortcite{Varian74}. Given a set of agents $N$, a set of goods $C$, and agent preferences over bundles of goods $\succ$, a CEEI mechanism finds a budget $b \in \mathbb{R}_+$ and price vector $p^* \in \mathbb{R}_+^{|C|}$, such that if each agent is allocated its favorite bundle of goods that costs no more than $b$, then each good in $C$ is allocated to exactly one agent in $N$, or divided in fractions summing to $1$ among the $N$.
In combinatorial allocation problems, such as course allocation, goods (seats in a class) are not divisible, and certain bundles of goods (class schedules) are not allowed to be assigned to an agent. As a result, an exact market clearing tuple $(b, p^*)$ may not exist.
To deal with this difficulty, CEEI was relaxed by Budish~\shortcite{budish2011combinatorial} to an approximate version, termed A-CEEI.
A-CEEI works by assigning nearly equal budgets to all agents, then searching for an approximately market clearing price vector and returning the allocation induced by those prices. The result may not clear the market exactly, but there is an upper bound on the worst-case market clearing error. The resulting allocation satisfies an approximate form of \emph{envy-freeness}~\cite{budish2011combinatorial}.

Both CEEI and A-CEEI take advantage of the dichotomy between agents and items which agents demand.
This makes our setting distinct: agents' demand in team formation is over subsets of other agents.
A technical consequence is that this gives rise to a hard constraint for CEEI that if an agent $i$ is paired with agent $j$, than $j$ must also be paired (assigned to) agent $i$; any relaxation of this constraint fails to yield a partition on the agents and consequently does not result in an admissible mechanism.
We therefore design an approximation of CEEI, termed A-CEEI-TF, that accounts for the specific peculiarities of our setting.
Conceptually, the A-CEEI-TF mechanism works by alternating between two steps. First, it searches in price space for approximate relaxed market-clearing prices. Second, it assigns a randomly selected unmatched agent to form a team with its favorite bundle of free agents that is affordable, based on current prices. The result is a mechanism that is strategyproof-in-the-large, and more fair than random serial dictatorship.

\begin{algorithm}
\begin{algorithmic}[1]
\REQUIRE {$(N, \succ, \underline{k}, \overline{k})$}
\STATE Randomly assign approximately equal budgets $b_i$ to the agents, $b_i \in [1, \bar{b}]$, $\bar{b} < 1 + 1 / |N|$. 
\STATE Randomly order the agents.
\STATE Search for a price vector $p$ in price space $\mathcal{P} = [0, \bar{b}]^{|N'|}$ that approximately clears the (relaxed) market among the $N'$ remaining agents, given agent budgets $b$.
\STATE Take the next unmatched agent in the random order, and assign it to join its favorite bundle of other free agents that it can afford at the current prices, and that leaves a feasible subproblem---i.e., feasible $(N', \underline{k}, \overline{k})$. If the agent cannot afford any remaining bundle of legal size that leaves a feasible subproblem, the agent is assigned its favorite remaining bundle of legal size that leaves a feasible subproblem.
\STATE Repeat steps 3 and 4 until each agent is on a team.
\end{algorithmic}
\caption{A-CEEI-TF Algorithm Outline.}
\label{A:aceei}
\end{algorithm}


A key part of A-CEEI-TF is a \emph{price update function}, which reflects the constraints of the team formation problem. 
We use a t\^{a}tonnement-like price update function $f$ in an auxiliary price space $\mathcal{\tilde{P}} = [-1, 1 + \bar{b}]^{|N'|}$, where $N'$ is the set of agents remaining (unassigned) at an iteration of the algorithm, and $\bar{b}$ is the supremum of allowable agent budgets. 
We make two requirements of a price update function, one ensuring that the iterative updates are well-defined, another to ensure that fixed points of the process are actual solutions.
\begin{definition}
A price update function $f$ is \emph{admissible} if (a) its fixed points correspond to (relaxed) market clearing, and (b) $\mathcal{\tilde{P}}$ is closed under $f$.
\end{definition}
We now define a candidate price update function, $f_{TF}$:
\begin{equation}
f_{TF}(\tilde{p})_j = t(\tilde{p})_j + \frac{(1 + \epsilon - (\epsilon / \bar{b}) t(\tilde{p})_j) D_j - U_j}{|N'|} 
\end{equation}
where $D_j$ is the number of agents that demand $j$ but whom $j$ does not demand, $U_j = 1$ if and only if no other agent demands $j$, and $0$ otherwise, $\bar{b}$ is the supremum of allowable agent budgets, and $t(\cdot)$ is a truncation function, which takes a price vector $\tilde{p}$ and truncates it to the $[0, \bar{b}]$ interval.
\begin{proposition}
\label{T:priceUpdate}
$f_{TF}(\cdot)$ is admissible.
\end{proposition}

While admissibility of $f_{TF}(\cdot)$ alone does not guarantee convergence of the iterative process, it does guarantee that if convergence happens, we have a solution.
The following proposition characterizes some of the properties such solutions possess.
\begin{proposition}
\label{T:aceeiTFic}
A-CEEI-TF is strategyproof-in-the-large. In addition, if A-CEEI-TF yields exact market clearing and induces the same allocation at each stage of price search, it yields envy bounded by a single teammate.
\end{proposition}

\subsubsection{Proof}
We sketch a proof of strategyproofness-in-the-large.

If a team formation problem is modified such that each agent is replaced with a measure-one set of copies of itself, each copy being measure zero, we arrive at what is called a continuum economy. If we run A-CEEI-TF in the continuum economy, any individual agent, being zero-measure, has no influence on the approximate equilibrium price vector arrived at by update function $f_{TF}(\cdot)$, at any iteration of the A-CEEI-TF mechanism. Therefore, the only effect the agent can have on the outcome is that, if the agent is randomly selected to choose its favorite affordable team of available agents that leaves a feasible subproblem, the agent's reported preferences determine which team the agent is assigned. Thus, it is a dominant strategy for the agent to report its true preferences, so that in this case the agent will be assigned its most-preferred allowable team. \qed

%% file: experiments.tex
\section{Experiments}

Although RSD and A-CEEI-TF possess desirable theoretical properties, these results are loose, and the only approximate fairness guarantee, shown for A-CEEI-TF, requires strong assumptions on the environment.
We now assess all of the proposed mechanisms empirically through simulations based both on randomly generated classes of preferences, as well as real-world data.
Our empirical results turn out to be both one-sided (if one is interested in achieving all three desired properties) and surprising: OPOP, a mechanism with no provable theoretical guarantees, tends to outperform others in fairness, and to perform nearly as well as the best other mechanism in truthfulness and social welfare.

\smallskip
\noindent\textbf{Data Sets: }
We use both randomly generated data and data from prior studies on preferences of human subjects over each other: 
\begin{itemize}
\item \textbf{Random-similar (R-sim)~\cite{othman2010finding}: } 
Each agent $i$, $i \in \{1, 2, \ldots, |N|\}$, is assigned the public value $i$. A private error term is added to the public value of $i$ to derive the value of $i$ to each other agent $j$, drawn independently from a normal distribution with zero mean and standard deviation $|N| / 5$. The private error term is redrawn until the sum of the private error term and public value is non-negative. Then the value of $i$ to $j$ is the sum of $i$ and the private error term. 
\item \textbf{Random-scattered (R-sca): } In this data set class, the value of a player $i$ is generated independently by each other player $j$. To determine the value of other players to player $j$, a total value of $100$ is divided at random among the other players as follows. Uniformly random numbers $\in [0, 100]$ are taken, to divide the region into $|N| - 1$ regions. The random draws for agent $j$ are sorted, producing $|N| - 1$ values for the other agents, as the differences between consecutive draws in sorted order. 
\item \textbf{Newfrat: } This data set comes from a widely cited study by Newcomb, in which $17$ students at the University of Michigan in 1956 ranked each other in terms of friendship ties. We use the data set from the final week, NEWC15 \cite{newcomb61fraternity}. We let $\underline{k} = 4$, $\overline{k} = 5$.
\item \textbf{Freeman: } The data are from a study of email messages sent among $32$ researchers in 1978. We use the third matrix of values from the study. The data show how many emails each researcher sent to each other during the study, which we use as a proxy for the strength of directed social links \cite{freeman79eies}. We let $\underline{k} = 5$, $\overline{k} = 6$.
\end{itemize}
For the randomly generated data sets, we set $|N| = 20$ and $\underline{k} = \overline{k} = 5$, and our results are averaged over $20$ generated preference rankings for all players.

The four classes of data set we analyze differ most saliently in their number of agents, $|N|$, and in the degree of similarity among the player preferences. For example, Random-similar agents largely agree on which other agents are most valuable, while Random-scattered agents have little agreement. Differences in degree of preference similarity lead to marked differences in the performance outcomes of the various mechanisms.

To measure agent preference similarity in a data set, we let $\mathcal{C}$ equal the mean cosine similarity among all pairs of distinct agents in the data set. Each agent assigns itself a value of $0$ or undefined, so we take the cosine similarity between agents $i$ and $j$ only over their values for agents in $N \setminus \{i, j\}$:
\begin{align*}
u_{i-ij} = u_i \setminus \{u_{ii}, u_{ij}\} \\
\mathcal{C} = \frac{ \sum_{i=1}^{|N|} \sum_{j=i+1}^{|N|} \frac{u_{i-ij} \cdot u_{j-ij}}{\|u_{i-ij}\| \|u_{j-ij}\|} }{ (|N|^2 - |N|) / 2 } \\
\end{align*}

In Table~\ref{T:randomCorr}, we present the mean cosine similarity for each data set we discuss in this paper. Higher cosine similarities indicate greater agreement among agents about the relative values of other agents. We also show the number of agents in each data set.
\begin{table}[h!]
\centering
\begin{tabular}{ l | c c c c}
     & $\mathcal{C}$ & $|N|$ & $\underline{k}$ & $\overline{k}$ \\
  \hline
  Random-similar 20 & 0.914 & 20 &  5 & 5 \\
  Random-scattered 20 & 0.499 & 20 & 5 & 5 \\
  Newfrat & 0.877 & 17 & 4 & 5 \\
  Freeman & 0.551 & 32 & 5 & 6 \\
\end{tabular}
\caption{Mean cosine similarity over all pairs of distinct agents; number of agents; minimum team size; and maximum team size. For random data set classes, $\mathcal{C}$ as shown is the mean over $20$ randomly generated instances of the class.}
\label{T:randomCorr}
\end{table}

\smallskip
\noindent\textbf{Empirical Analysis of Incentive Compatibility: }
To study the incentive compatibility of the mechanisms, we used a protocol similar to that used by Vorobeychik and Engel for estimating the regret of a strategy profile \cite{vorobeychik2011average}. We ran each mechanism on $8$ versions of each data set, with different random orders over the players, which we held in common across data sets. 
We generate deviations from truthful reporting for agent $j$ one at a time until $25$ unique deviations have been produced. 
To produce a deviation from an agent's truthful values for other agents, we first randomly select a number of pairs of values to swap according to a Poisson distribution with $\lambda = 1$, with $1$ added. For each pair of values to swap, we first select the rank of one of them, with lower (better) ranks more likely. 

The results are shown in Table~\ref{T:regretMean}.\footnote{We only report results for the two random data sets and Newfrat, as it was not feasible to rigorously analyze regret for the far larger Freeman data set.
(However, the Freeman data set is similar to Random-scattered; see supplemental material for details.)}

RSD is not shown, since it is provably strategyproof, but, remarkably, A-CEEI-TF empirically produces higher (worse) regret of truthful reporting than HBS or OPOP, even though A-CEEI-TF is strategyproof-in-the-large, and the others are not.
Both HBS and OPOP appear to offer players only small incentives to lie, with HBS slightly better.

\begin{table}[h!]
\centering
\begin{tabular}{ l | c c c}
     & R-sim. & R-sca. & Newfrat\\
  \hline
  HBS & $0.02\pm 0.02$ & $0.04\pm 0.02$ & $0.03\pm 0.02$\\
  OPOP & $0.07\pm 0.02$ & $0.10 \pm 0.05$ & $0.06 \pm 0.02$\\
  A-CEEI-TF & $0.19 \pm 0.04$ & $0.29 \pm 0.07$ & $0.19 \pm 0.03$\\
 Max-welfare & $0.20 \pm 0.03$ & $0.29 \pm 0.07$ & $0.22 \pm 0.03$\\
\end{tabular}
\caption{Mean maximum observed regret of truthful reporting, with $95\%$ confidence intervals.
}
\label{T:regretMean}
\end{table}

\leaveout{

\begin{table}[h!]
\centering
\begin{tabular}{ l | c }
     & Newfrat \\
  \hline
  HBS & $0.030\pm 0.018$ \\ 
  OPOP & $0.057 \pm 0.019$ \\ 
  A-CEEI-TF & $0.188 \pm 0.032$ \\ 
  Max-welfare & $0.216 \pm 0.026$
\end{tabular}
\caption{Mean across runs of maximum observed regret of truthful reporting, meaning the maximum fraction of total utility gained from misreporting preferences. $0$ would be best, $1$ worst. Results are over 8 runs with 25 deviations per agent per run. Data are from the Newfrat data set. $95\%$ confidence intervals are shown, based on the t-distribution.}
\label{T:regretMeanNewfrat}
\end{table}
}

\smallskip
\noindent\textbf{Social Welfare: }
To facilitate comparison, we normalize the total utility of all teammates for each agent to $1$, so that social welfare (already normalized for the number of players) falls in the $[0,1]$ interval.
For comparison, we also include optimal social welfare for both Random data sets, as well as Newfrat.\footnote{It was infeasible to compute this for the Freeman data set due to its size.}

The only related theoretical result is that RSD is ex post Pareto optimal; the other three mechanisms do not even possess this guarantee.
This makes our results, shown in Tables~\ref{T:sw} and~\ref{T:nfsw}, remarkable: on Random-similar and Newfrat data sets (both with preferences relatively similar across players), there is little difference in welfare generated by the different mechanisms, but on Random-scattered and Freeman data sets, OPOP statistically significantly outperforms the others.

\begin{table}[h!]
\centering
\begin{tabular}{ l | c c }
     & R-sim. & R-sca. \\
  \hline
  RSD & $0.22\pm 0.004$ & $0.25 \pm 0.01$ \\
  A-CEEI-TF & $0.22 \pm 0.004$ & $0.25\pm 0.01$\\
  HBS & $0.22 \pm 0.004$ & $0.25 \pm 0.02$\\
  OPOP & $0.22 \pm 0.003$ & $0.27 \pm 0.01$\\
  Max-welfare & $0.25 \pm 0.001$ & $0.35 \pm 0.01$\\
\end{tabular}
\caption{Mean social welfare for the two Random data sets, with $95\%$ confidence intervals. 
}
\label{T:sw}
\end{table}

\begin{table}[h!]
\centering
\begin{tabular}{ l | c c }
     & Newfrat & Freeman \\
  \hline
  RSD & $0.23 \pm 0.01$ & $0.20 \pm 0.01$\\
  A-CEEI-TF & $0.23 \pm 0.01$ & $0.19 \pm 0.01$\\
  HBS & $0.22 \pm 0.05$ & $0.20 \pm 0.01$\\
  OPOP & $0.22 \pm 0.05$ & $0.24 \pm 0.02$\\
  Max-welfare & $0.27 \pm 0.00$ & -\\
\end{tabular}
\caption{Mean social welfare for the Newfrat and Freeman data sets, with $95\%$ confidence intervals.
}
\label{T:nfsw}
\end{table}

\smallskip
\noindent\textbf{Fairness: }
Fairness of an allocation (in our case, a partition of players) can be conceptually described as the relative utility of best- and worst-off agents. Formally, we measure fairness in the experiments as the fraction of agents whose envy is bounded by a single teammate (as defined above).

Our fairness results, shown in Tables~\ref{T:envy} and~\ref{T:envynf} are unambiguous: RSD is always worse, typically by a significant margin, then the other mechanisms.
This is intuitive, and is precisely the reason why alternatives to RSD are commonly considered.
What is far more surprising is that A-CEEI-TF, in spite of some theoretical promise on the fairness front, and in spite of being explicitly designed for fairness, is in all but one case the \emph{second worst}.
While HBS and OPOP are comparable on the high-similarity data sets (Random-similar and Newfrat), it \emph{dominates} all others on the dissimilar data sets (Random-scattered and Freeman).

\begin{table}[h!]
\centering
\begin{tabular}{ l | c c }
     & R-sim. & R-sca. \\
  \hline
  RSD & $0.43 \pm 0.03$ & $0.59 \pm 0.05$ \\
  A-CEEI-TF & $0.66 \pm 0.05$ & $0.62 \pm 0.05$ \\
  HBS & $0.71 \pm 0.04$ & $0.61 \pm 0.06$ \\
  OPOP & $0.70 \pm 0.06$ & $0.79 \pm 0.04$ \\
\end{tabular}
\caption{Mean fraction of agents with envy bounded by a single teammate for the two Random data sets, with $95\%$ confidence intervals. 
}
\label{T:envy}
\end{table}

\begin{table}[h!]
\centering
\begin{tabular}{ l | c c }
     & Newfrat & Freeman \\
  \hline
  RSD & $0.36 \pm 0.02$ & $0.43 \pm 0.05$ \\
  A-CEEI-TF & $0.57 \pm 0.03$ & $0.55 \pm 0.05$\\
  HBS & $0.67 \pm 0.05$ & $0.64 \pm 0.04$ \\
  OPOP & $0.68 \pm 0.07$ & $0.78 \pm 0.05$\\
\end{tabular}
\caption{Mean fraction of agents with envy bounded by a single teammate for the Newfrat and Freeman data sets, with $95\%$ confidence intervals.
}
\label{T:envynf}
\end{table}

Next, we consider informal perspectives on the fairness of the various mechanisms.
For mechanisms that use a random serial order over players, we can study typical outcomes for a player given its serial index. If players with lower (better) serial indexes receive drastically better outcomes than agents with higher (worse) indexes, such a mechanism is not very fair. In Figure~\ref{F:fairness} (left), we plot a smoothed version of the mean fraction of total utility achieved by agents at each random serial index from $0$ to $19$, for the mechanisms RSD, HBS draft, OPOP draft, and A-CEEI-TF. Results are based on $20$ instances of Random-scattered preferences, held in common across the mechanisms, with different serial orders over the players.
From this Figure, it is apparent that random serial dictatorship gives far better outcomes to the best-ranked agents than to any others. 
Surprisingly, A-CEEI-TF follows a similar pattern in this case, although we did observe that for some other game types (not shown), A-CEEI-TF's curve gives better outcomes to low-ranked agents than RSD. 
In the HBS draft, a ``shelf'' of high utility for the several best-ranked agents is typical, as all of the team captains receive similarly high utility, with a steep drop-off in utility for non-captain agents. 
In the OPOP draft, in contrast to all others, the utility curve is far more flat across random serial indexes: even the agents with high (bad) serial indexes achieve moderately good outcomes for themselves. 

\noindent
\begin{figure}[h!]
\centering
\begin{tabular}{cc}
\includegraphics[width=1.5in]{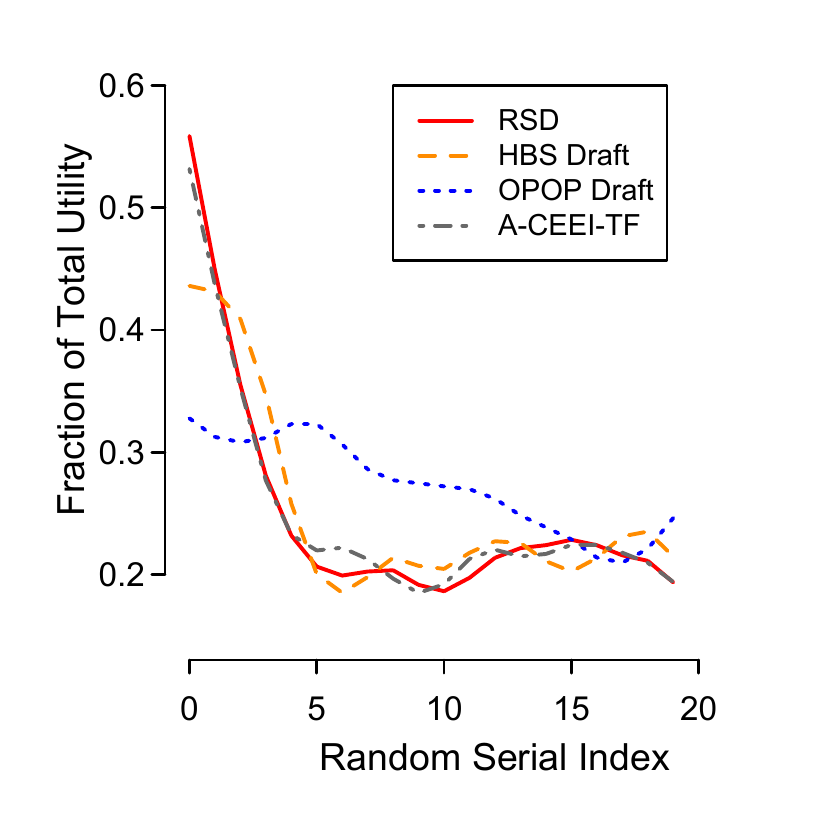} & 
\includegraphics[width=1.5in]{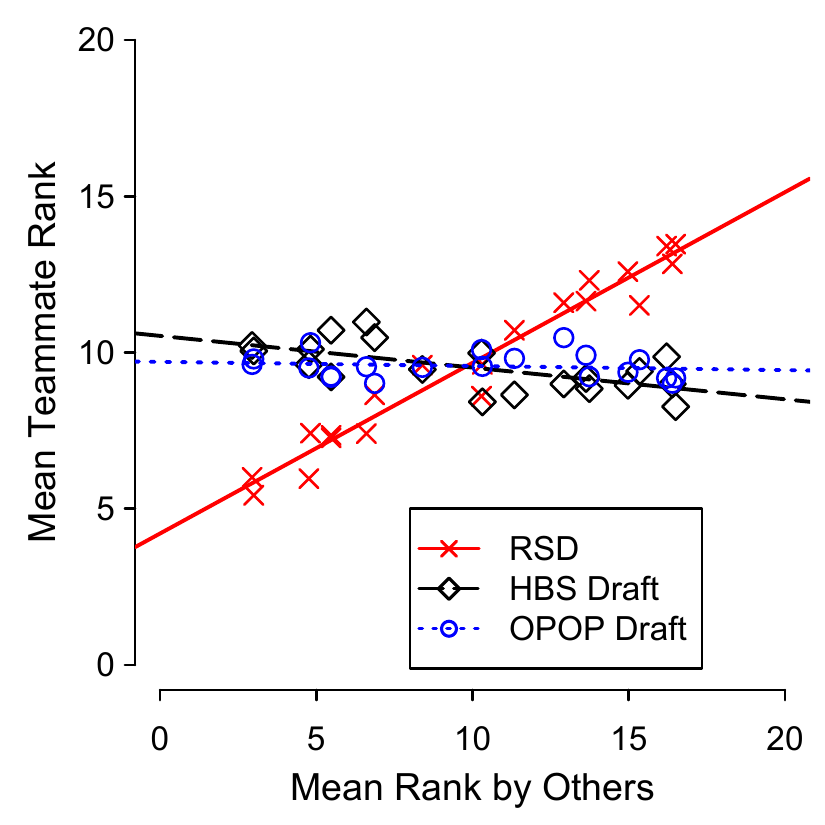}
\end{tabular}
\caption{Left: Mean fraction of total utility earned versus the random serial index of the agent. Fraction of total utility is between $0$ (worst) and $1$ (best). A cubic smoothing spline is applied.  Right: Mean rank of an agent's teammates, versus mean rank of the agent by other agents. Possible ranks range from $1$ (best) to $20$ (worst). Each point represents a single agent's mean outcome. Best-fit lines use ordinary least squares.}
\label{F:fairness}
\end{figure}

Some mechanisms for team formation tend to give better outcomes to an agent that is ``popular,'' having a high mean value to the other agents. For example, random serial dictatorship biases outcomes in favor of popular players, because even if a popular player is not a team captain, it is likely that this player will be selected by some team captain along with other desirable players. An unpopular player, however, will likely be left until near the final iteration of RSD, to be selected along with other unpopular players. Therefore, we might expect RSD to yield better outcomes to popular players, especially when agents' preferences are highly similar.
To quantify this intuition, we plot in Figure~\ref{F:fairness} (right) the mean rank of an agent's teammates according to the agent's preferences, versus the agent's mean rank assigned by the other agents. Each point in the scatter plot represents a single agent's mean outcomes across $20$ instances of Random-similar preferences, held in common across the mechanisms. We find best-fit lines via OLS regression, for each of RSD, HBS draft, and OPOP draft. The results indicate that, as expected, RSD offers better outcomes to popular agents than to unpopular ones, with a distinctly positive trend line. The HBS draft and OPOP draft appear less biased for or against popular agents, with OPOP showing slightly lower correlation than HBS between an agent's popularity and the mean value of its assigned team.

%% file: conclusion.tex
\section{Conclusion}

We considered team formation as a mechanism design problem, in which the mechanism elicits agents' preferences over potential teammates in order to partition the agents into teams. The teams produced should have high social welfare and fairness, in the sense that few agents should prefer to switch teams with others. We proposed two novel mechanisms for this problem: a version of approximate competitive equilibrium for equal incomes (A-CEEI-TF), and the one-player-one-pick draft (OPOP). We showed theoretically that A-CEEI-TF is strategyproof-in-the-large and approximates envy-freeness. OPOP lacks these theoretical guarantees but empirically outperformed A-CEEI-TF in truthfulness and fairness, as well as in social welfare for data sets with sufficiently dissimilar agent preferences. In addition, OPOP surpassed other mechanisms tested, including random serial dictatorship and the HBS draft, in social welfare and fairness. The HBS draft, however, produced slightly better truthfulness that the OPOP draft. Given the relative simplicity of implementing OPOP, this mechanism emerges as a strong candidate for team formation settings.

%% file: appendix.tex
\section{Appendix}

\subsection{Team Size Constraints that Admit a Feasible Partition}

Some tuples $(N, \underline{k}, \overline{k})$ are not feasible, meaning that it is not possible to divide $|N|$ players into teams with sizes in $[\underline{k}, \overline{k}]$. Recall that we require by definition $1 \leq \underline{k} \leq \overline{k} \leq |N|$. For a minimal example, it is not possible to divide $3$ players into teams with $\underline{k} = 2$, $\overline{k} = 2$.


\begin{observation}
\label{T:feasibility}
A tuple $(N, \underline{k}, \overline{k})$ is feasible if and only if: $\underline{k}$ divides $|N|$, $\overline{k}$ divides $|N|$, or $(|N| \setminus \underline{k}) > (|N| \setminus \overline{k})$, where $\setminus$ signifies integer division.
\end{observation}

\subsection{Computing A Social Welfare-Maximizing Partition}

In order to find a social welfare-maximizing partition, we formulate the problem as an MIP and solve it using CPLEX. We are given a feasible team formation problem as $(N, \succ, \underline{k}, \overline{k})$, where either $\underline{k} = \overline{k}$, or $\underline{k} + 1 = \overline{k}$ and neither $\underline{k}$ nor $\overline{k}$ divides $|N|$. Let $T$ equal the number of teams that results when as many teams of size $\overline{k}$ as possible are formed, the rest being of size $\underline{k}$. 

We introduce a matrix $x \in  \{0, 1\}^{T \times |N|}$, where each row corresponds to one team, and the $1$ values in the row indicate which agents are on that team. We also introduce a dummy variable, $S \in \{0, 1\}^{T \times |N| \times |N|}$, where $S_{tij} = 1$ if agents $i$ and $j$ are both on team $t$, otherwise $0$. $S'_{ij} \in \{0, 1\}^{|N| \times |N|}$ is $1$ if and only if agents $i$ and $j$ are on the same team, and is the sum over $T$ values of $S_{tij}$ for $i$ and $j$.

The MIP for maximize social welfare is then:

\begin{align*}
\max_x : \sum_{i=1}^{|N|} \sum_{j=i+1}^{|N|} S'_{ij} (u_{ij} + u_{ji}) \\
\textrm{subject to:} \\
x \in \{0, 1\}^{T \times |N|} \\
S \in \{0, 1\}^{T \times |N| \times |N|} \\
S' \in \{0, 1\}^{|N| \times |N|} \\
2 S_{tij} \leq x_{ti} + x_{tj} : \forall t, i, j \\
S_{tij} \geq x_{ti} + x_{tj} - 1 : \forall t, i, j \\
S'_{ij} = \sum_{t=1}^T S_{tij} : \forall i, j \\
\underline{k} \leq \sum_{j=1}^{|N|} x_{tj} \leq \overline{k} : \forall t \\
\sum_{t=1}^{T} x_{tj} = 1 : \forall j \\
\end{align*}

We implemented a second version of the maximize social welfare MIP, which appears to run markedly faster. In this version, a matrix $x \in \{0, 1\}^{|N| \times |N|}$ has a $1$ in row $i$ for each agent on the team of agent $i$. We require that $x_{ij} = x_{ji}$ so that demands are reciprocal. We also require that each agent demand itself, so $x_{ii} = 1$. Each team size must be in $[\underline{k}, \overline{k}]$, so the sum of each row of $x$ must be in this range. Finally, we require that any two rows in $x$ either not have a $1$ in any of the same columns, or must be identical; this means that if some agent $i$ appears on two teams (i.e., in two rows), those teams must contain exactly the same agents.

\begin{align*}
\max_x: \sum_{i=1}^{|N|} \sum_{j=1}^{|N|} x_{ij} u_{ij} \\
\textrm{subject to:} \\
x \in \{0, 1\}^{|N| \times |N|} \\
x_{ii} = 1 : \forall i \\
x_{ij} - x_{ji} = 0 : \forall i, j \\
\underline{k} \leq \sum_{j=1}^{|N|} x_{ij} \leq \overline{k} : \forall i \\
x_{ij} + x_{i'j} + x_{ij'} - x_{i'j'} \leq 2 : \forall i, i' \neq i, j, j' \neq j \\
\end{align*}

\smallskip
\noindent
\subsection{Proof of Proposition~\ref{T:envyBoundedBySingle}} Consider a team formation problem with $6$ agents, $\{A, B, C, D, E, F\}$, $\underline{k} = 3$, $\overline{k} = 3$, so that two equal-size teams must be formed. The agents' additive separable preferences are encoded in Table \ref{tab:envyBoundSingle}.

\begin{table}[h!]
\centering
\begin{tabular}{ c | c c c c c c c }
     & $A$ & $B$ & $C$ & $D$ & $E$ & $F$ \\
  \hline
  $A$ & x & 0 & 1 & 2 & 4 & 8 \\
  $B$ & 8 & x & 4 & 2 & 1 & 0 \\
  $C$ & 8 & 0 & x & 4 & 2 & 1 \\
  $D$ & 8 & 1 & 0 & x & 4 & 2 \\
  $E$ & 8 & 2 & 1  & 0 & x & 4 \\
  $F$ & 8 & 4 & 2 & 1 & 0 & x \\
\end{tabular}
\caption{Each row $i$ encodes the additive separable value for agent $i$ of each other agent.}
\label{tab:envyBoundSingle}
\end{table}

No partition of these agents into two teams of size $3$ gives every agent envy bounded by a single teammate. To see this, consider that each agent other than $A$ has a bliss point on a team with $A$ and one other agent, where the second agent is $C$ for agent $B$, $D$ for agent $C$, and so on until ``wrapping around'' with $B$ for agent $F$. Three of the agents will not be on a team with agent $A$, and at least one of these agents, say agent $i$, will not be on a team with its second-favorite agent either. Some other agent $j$ must then be on a team with the two most-preferred agents of the player $i$. By construction, player $i$ is on a team of value $3$ or less, while the team of agent $j$ has value $12$ to agent $i$, and value $4$ to agent $i$ with its more valuable player (player $A$) removed. Therefore, envy cannot be bounded by a single teammate for all agents. \qed

\paragraph{Maximin Share Guarantee}

The maximin share guarantee is a concept from multi-unit assignment that can be applied to hedonic games; in multi-unit assignment, ``bundles'' of ``items'' are allocated to agents instead of ``teams.'' A maximin share for an agent is the agent's least-preferred bundle in a \emph{maximin split} for the agent.  A maximin split for agent $i$ is a partition of all items into bundles such that each agent can receive one bundle, where the partition maximizes the utility to $i$ of the least-valuable bundle. Note that envy-freeness in a split for an agent implies proportionality, and proportionality implies maximin shares \cite{procaccia2014fair}.

We introduce a modified version of the maximin share guarantee in our setting, the \emph{maximin share guarantee for team formation}. A partition in a team formation problem provides maximin shares for team formation, if each agent weakly prefers its team to its maximin share. Given $(N, \succ_i, \underline{k}, \overline{k})$, consider all partitions of the agents $N$ with team sizes in $[\underline{k}, \overline{k}]$. A maximin split for agent $i$ is any such partition that maximizes the value for $i$ of the least-preferred team containing $i$, which would result from swapping $i$ with some agent $j$, where $j$ may equal $i$. A maximin share is the least-preferred team in a maximin split. This definition of the maximin share guarantee preserves a useful property of the guarantee for multi-unit assignment, which is that an agent may be assigned one of $|N|$ teams from its maximin split (up to the loss of the other agent in the swap from a team, in the case of team formation).

The maximin share guarantee for team formation, in the case of mechanism design for team formation problems (MDTFs) with non-negative values, is trivially satisfiable for $\overline{k} = |N|$: If the grand coalition is assigned, every agent's payoff is maximized, so envy-freeness is achieved. The maximin share guarantee for team formation, in the case of MDTFs with non-negative values, is also trivially satisfiable for $\overline{k} = 2$, where $|N|$ is even. In the maximin split for an agent with $\overline{k} = 2$, all agents are grouped in pairs (assuming $|N|$ is even), and the agent's maximin share is the pair with its least-favorite other agent. Thus, any grouping of the agents into pairs satisfies the maximin share guarantee.

\smallskip
\begin{proposition}
\label{T:maximin}
For some MDTFs with non-negative values, $\overline{k} \geq 3$, and $\overline{k} < |N|$, the maximin share guarantee for team formation is not satisfiable. As a result, proportionality and envy-freeness are not satisfiable either in such cases.
\end{proposition}
\begin{proof}
Consider a team formation problem with $6$ agents, $\{A, B, C, D, E, F\}$, $\underline{k} = 3$, $\overline{k} = 3$, so that two equal-size teams must be formed. The agents' additive separable preferences are encoded in Table \ref{tab:maximin}.

\begin{table}[h!]
\centering
\begin{tabular}{ c | c c c c c c c }
     & $A$ & $B$ & $C$ & $D$ & $E$ & $F$ \\
  \hline
  $A$ & x & 1 & 0 & 2 & 3 & 2 \\
  $B$ & 2 & x & 1 & 0 & 3 & 2 \\
  $C$ & 2 & 2 & x & 1 & 0 & 3 \\
  $D$ & 1 & 2 & 2 & x & 2 & 1 \\
  $E$ & 0 & 2 & 3  & 1 & x & 2 \\
  $F$ & 2 & 1 & 2 & 2 & 1 & x \\
\end{tabular}
\caption{Each row $i$ encodes the additive separable value for agent $i$ of each other agent.}
\label{tab:maximin}
\end{table}

\smallskip
It is easy to see that in the problem described above, each agent's maximin share for team formation has value $3$. Furthermore, any of the $10$ possible partitions of the agents into two equal-size teams leaves some agent with utility of $2$ or less. Therefore, no partition for this game provides every agent with a maximin share. In consequence, no partition for this problem is proportional or envy-free either. 
\end{proof}

\subsection{Proof of Proposition~\ref{T:rsdic}}

Random serial dictatorship for team formation is strategyproof, meaning that it is a dominant strategy for each agent to report its true values for other agents, regardless of the other agents' reports. Consider that each agent affects the partition returned by RSD only if the agent is a team captain, based on the serial order over players and the choices of earlier team captains. Therefore, if an agent is not a team captain, the agent's report makes no difference. If an agent is a team captain, the agent is assigned its most-preferred team based on the remaining players, team size constraints, and the agent's reported preferences. Therefore, it is a dominant strategy for the agent to report its true preferences.

Random serial dictatorship is ex post Pareto efficient for MDTFs with strict preferences over teams (i.e., without indifferences over teams). The first team captain is assigned its strictly most-preferred team, of the largest feasible team size. Thus, for non-negatively valued MDTFs, no other agent could swap or take agents from the first captain's team without decreasing its value to the captain. The same argument holds, by induction, for later team captains. Thus, the partition returned by RSD is ex post Pareto efficient.

Random serial dictatorship for team formation is not envy-free, because for some problem instances, no envy-free partition exists, as shown in Proposition~\ref{T:maximin}. \qed

\subsection{Proof of Proposition~\ref{T:hbs}}

The Harvard Business School draft for team formation is not envy-free, because for some problem instances, no envy-free partition exists, as shown in Proposition~\ref{T:maximin}.

The HBS draft for team formation is not ex post Pareto efficient. For a counter-example, consider the following team formation problem, with $\underline{k} = \overline{k} = 3$ and player order $(A, B, C, D, E, F)$, such that $A$ and $B$ will be the team captains. The agent preferences are encoded in Table \ref{tab:hbsEfficient}.

\begin{table}[h!]
\centering
\begin{tabular}{ c | c c c c c c c }
     & $A$ & $B$ & $C$ & $D$ & $E$ & $F$ \\
  \hline
  $A$ & x & 1 & 8 & 0 & 6 & 4 \\
  $B$ & 1 & x & 10 & 0 & 5 & 3 \\
  $C$ & 0 & 8 & x & 5 & 4 & 2 \\
  $D$ & 0 & 8 & 5 & x & 4 & 2 \\
  $E$ & 8 & 0 & 5  & 4 & x & 2 \\
  $F$ & 8 & 0 & 5 & 4 & 2 & x \\
\end{tabular}
\caption{Each row $i$ encodes the additive separable value for agent $i$ of each other agent.}
\label{tab:hbsEfficient}
\end{table}

In this example, the following selections are made in the HBS draft. $A$ selects $C$, $B$ selects $E$, $B$ selects $F$, and then $A$ selects $D$. The resulting partition is $\{(ACD), (BEF)\}$. This partition is not ex post Pareto efficient, because all the agents would prefer the partition $\{(AEF), (BCD)\}$.

The HBS draft for team formation is not strategyproof. Consider a minimal-size example of a problem instance with non-negative values, where the HBS draft is not strategyproof. Let $N = \{A, B, C, D, E, F\}$, $\underline{k} = 3$, and $\overline{k} = 3$. The players' additive separable values to each other are encoded in Table \ref{tab:hbsIc}.

\begin{table}[h!]
\centering
\begin{tabular}{ c | c c c c c c c }
     & $A$ & $B$ & $C$ & $D$ & $E$ & $F$ \\
  \hline
  $A$ & x & 5.0 & 4.9 & 7.0 & 0.2 & 0.0 \\
  $A'$ & x & 5.0 & 6.0 & 7.0 & 0.2 & 0.0 \\
  $B$ & 0.0 & x & 1.1 & 1.6 & 1.2 & 1.3 \\
  $C$ & 0.0 & 1.1 & x & 1.6 & 1.2 & 1.3 \\
  $D$ & 0.0 & 1.1 & 1.6 & x & 1.2 & 1.3 \\
  $E$ & 0.0 & 1.1 & 1.6  & 1.2 & x & 1.3 \\
  $F$ & 0.0 & 1.1 & 1.6 & 1.2 & 1.3 & x \\
\end{tabular}
\caption{Each row $i$ encodes the additive separable value for agent $i$ of each other agent. Row $A'$ shows the false report of agent $A$.}
\label{tab:hbsIc}
\end{table}

In the problem shown in Table \ref{tab:hbsIc}, agent A gets greater utility by misreporting its preferences as in $A'$ for $24$ serial orders over the players, and never receives worse utility than when reporting $A$. In each of the $24$ cases where $A$ gains by reporting $A'$, utility improves from $5.2$ to $9.9$, as the received team shifts from $(ABE)$ to $(ABC)$. This occurs when the first player to choose is $A$ and the second is $D$. Therefore, the HBS draft is not strategyproof, because in the example shown, player $A$ gets better expected utility by misreporting its preferences. \qed

\subsection{Proof of Proposition~\ref{T:opopic}}

Consider a minimal-size example of a team formation problem with non-negative values, where the OPOP draft is not strategyproof. Let $N = \{A, B, C, D, E, F\}$, $\underline{k} = 3$, and $\overline{k} = 3$. The players' additive separable values to each other are encoded in Table \ref{tab:opopic}.

\begin{table}[h!]
\centering
\begin{tabular}{ c | c c c c c c c }
     & $A$ & $B$ & $C$ & $D$ & $E$ & $F$ \\
  \hline
  $A$ & x & 5.0 & 4.9 & 7.0 & 0.2 & 0.0 \\
  $A'$ & x & 5.0 & 6.0 & 7.0 & 0.2 & 0.0 \\
  $B$ & 0.0 & x & 1.1 & 1.6 & 1.2 & 1.3 \\
  $C$ & 0.0 & 1.1 & x & 1.6 & 1.2 & 1.3 \\
  $D$ & 0.0 & 1.1 & 1.6 & x & 1.2 & 1.3 \\
  $E$ & 0.0 & 1.1 & 1.6  & 1.2 & x & 1.3 \\
  $F$ & 0.0 & 1.1 & 1.6 & 1.2 & 1.3 & x \\
\end{tabular}
\caption{Each row $i$ encodes the additive separable value for agent $i$ of each other agent. Row $A'$ shows the false report of agent $A$.}
\label{tab:opopic}
\end{table}

\smallskip
The OPOP draft is a randomized mechanism, which works by uniformly randomly choosing a permutation of the players, then acting deterministically based on that order. Strategyproofness for a randomized mechanism means that a player's expected payoff is maximized by truthful reporting, regardless of other players' actions. To show that the OPOP draft is not strategyproof, it suffices to show that player $A$ gains in expected payoff, based on the true preferences in row $A$ of the table, by misreporting its preferences as in row $A'$ in the table, if other players report their preferences as displayed in the table.

The example problem has $6$ players, so there are $6! = 720$ permutations, and $720$ equally likely outcomes from the OPOP draft, some of which are identical in their induced partitions. It turns out that the payoff for player $A$ is better when reporting $A'$ instead of $A$ for $18$ of $720$ permutations, and worse in $6$ others. Specifically, player $A$ performs better by lying if and only if the order of the first three players to choose is $(ADB)$, $(ADE)$, or $(ADF)$. Player $A$ does worse by reporting $A'$ if and only if the order of the first three players is $(ADC)$. After factoring in the value of the resulting partitions to player $A$, it results that the expected gain from defecting to $A'$ is $0.08$ utils per game. Thus, the OPOP draft is not strategyproof.

The One-Player-One-Pick draft for MDTFs is not envy-free, because for some team formation problem instances, no envy-free partition exists, as shown in Proposition~\ref{T:maximin}. 

The OPOP draft for MDTFs is not ex post Pareto efficient, as shown in the following minimal example. Let $N = \{A, B, C, D, E, F\}$, $\underline{k} = 3$, and $\overline{k} = 3$. Let the serial order of the players be alphabetical, so $A$ selects first, and so on. The players' additive separable values to each other are encoded in Table \ref{tab:opopEfficient}.

\begin{table}[h!]
\centering
\begin{tabular}{ c | c c c c c c c }
     & $A$ & $B$ & $C$ & $D$ & $E$ & $F$ \\
  \hline
  $A$ & x & 2 & 10 & 9 & 6 & 0 \\
  $B$ & 0 & x & 10 & 9 & 2 & 6 \\
  $C$ & 0 & 10 & x & 2 & 6 & 9 \\
  $D$ & 10 & 0 & 2 & x & 6 & 9 \\
  $E$ & 10 & 0 & 2  & 6 & x & 9 \\
  $F$ & 0 & 10 & 2 & 6 & 9 & x \\
\end{tabular}
\caption{Each row $i$ encodes the additive separable value for agent $i$ of each other agent.}
\label{tab:opopEfficient}
\end{table}

In the example problem in Table \ref{tab:opopEfficient}, the following selections are made, in order. Team captain $A$ selects $C$, team captain $B$ selects $D$, $C$ (already teamed with $A$) selects $F$, and as a result $D$ (already teamed with $B$) must select $E$. This produces the partition $\{(ACF), (BDE)\}$. Every player would receive greater utility from the alternative partition $\{(ADE), (BCF)\}$. Therefore, the OPOP draft is not ex post Pareto efficient. \qed

\subsection{Details of A-CEEI-TF}

Each agent reports to the mechanism the additive separable utility it gains from each other agent. The mechanism uses these reports to derive which affordable team each agent demands, given each agent's budget and price. The total utility for an agent $i$ of a team of other agents of size in $[\underline{k} - 1, \overline{k} - 1]$, is the sum of the utilities for agent $i$ of the agents in that team. The clearinghouse will provisionally assign agent $i$ the team of highest utility for $i$ that it can afford, or the empty team if it cannot afford any team of legal size.

Just as in combinatorial matching, we cannot in general hope for an exact market clearing solution (and, consequently, an exact CEEI) in our setting:
\begin{proposition}
\label{T:noExactClearing}
There exist team formation settings where no price and budget vectors $(p, b)$ exist that induce exact market clearing.
\end{proposition}


\begin{proof}
We define \emph{exact market clearing for team formation} to mean that when each agent is allocated its favorite team that is affordable:

\begin{itemize}
\item Each agent is assigned to a team with size in $[\underline{k}, \overline{k}]$, including itself in the team size.
\item For any two distinct agents $i$ and $j$, $i$ demands $j$ if and only if $j$ demands $i$.
\end{itemize}

This result marks a distinction from the course allocation problem (also known as combinatorial assignment), where with sufficiently unequal budgets, some price vector must exist that induces exact market clearing \cite{budish2011combinatorial}.

\noindent
\emph{Example.} Consider a team formation problem with $4$ agents, $\{A, B, C, D\}$, with $\underline{k} = 2$ and $\overline{k} = 2$. The agents' additive separable preferences are encoded in Table \ref{tab:noExactClear}.

\begin{table}[h!]
\centering
\begin{tabular}{ c | c c c c }
     & $A$ & $B$ & $C$ & $D$ \\
  \hline
  $A$ & x & 2 & 1 & 0 \\
  $B$ & 0 & x & 2 & 1 \\
  $C$ & 1 & 0 & x & 2 \\
  $D$ & 2 & 1 & 0 & x \\
\end{tabular}
\caption{Each row $i$ encodes the additive separable value for agent $i$ of each other agent.}
\label{tab:noExactClear}
\end{table}

\smallskip
In the example in Table \ref{tab:noExactClear}, only two resulting partitions must be considered, due to symmetry: $\{\{A, B\}, \{C, D\}\}$ and $\{\{A, C\}, \{B, D\}\}$.

In order for the market to clear as $\{\{A, B\}, \{C, D\}\}$, each agent must be able to afford its assigned partner, and each agent must not be able to afford any other agent that it prefers. If we assume the market clears in this way, the resulting inferences produce an impossible conclusion:

\begin{align*}
p_A \leq b_B < p_C \leq b_D < p_A 
\end{align*}

A similar conclusion can be drawn from the other partition, $\{\{A, C\}, \{B, D\}\}$. We find that no $(b, p)$ induces exact market clearing for this problem, no matter how unequal the agents' budgets in $b$ are.
\end{proof}

Define the price update error as
\[
z(p)_j =(1 + \epsilon - (\epsilon / \bar{b}) p_j) D_j - U_j.
\]
Since exact market clearing is not feasible, we make use of a relaxation.
We define a notion of \emph{relaxed market clearing}, which we will use in our search for approximately market-clearing prices. This form of market clearing allows a player on the market to be under-demanded if and only if its price is $0$. We adapt this relaxed form of market clearing from Budish \cite{budish2011combinatorial}. In relaxed market clearing:

\begin{itemize}
\item Each agent is assigned to a team with size in $[\underline{k}, \overline{k}]$, including itself in the coalition size, unless the agent has price $0$, in which case it may alternatively be assigned the team containing only itself.
\item For any two distinct agents $i$ and $j$, $i$ demands $j$ if and only if $j$ demands $i$.
\end{itemize}


\begin{proposition}
\label{T:fixedPoint}
There exists a point $p^*$ in price space $\mathcal{P}$, for which a convex combination of the error terms $z(p')$ induced by some set of prices near $p^*$ equals the zero vector, as in relaxed market clearing.
\end{proposition}
\begin{proof}
The proposition can be proved through a straightforward adaptation of the proof of Theorem 1 in Budish \cite{budish2011combinatorial}. First we note that any fixed point of $f_{TF}(\cdot)$ induces relaxed market clearing. Furthermore, the price space $\mathcal{\tilde{P}}$ is compact, convex, and closed under $f_{TF}(\cdot)$. We introduce a function $F(\tilde{p})$ that takes the convex hull of all values $f_{TF}(\tilde{p}')$, for $\tilde{p}'$ approaching in the limit but not equal to $\tilde{p}$. 

Lemma 2.4 of Cromme and Diener is useful here \cite{cromme1991fixed}. Take any space $\mathcal{X}$ that is a non-empty, compact, convex subset of $\mathbb{R}^d$, and any function $f_{TF}: \mathcal{X} \rightarrow \mathcal{X}$, such that $\mathcal{X}$ is closed under $f_{TF}$. Define $H_f: \mathcal{X} \rightarrow S$, $S \subset \mathcal{X}$, to be the function that takes any $x$ in $\mathcal{X}$ and yields the set:

\begin{align*}
\{y \in \mathcal{X} | \exists \textrm{ a sequence } x_i \rightarrow x, x_i \neq x : f_{TF}(x_i) \rightarrow y \}
\end{align*}

Lemma 2.4 of Cromme and Diener states that for any such function $f_{TF}$, there exists a point $x^* \in \mathcal{X}$ such that $x^*$ is in the convex hull of $H_f(x^*)$.

We can apply Cromme and Diener's result to $F(\cdot)$, because the auxiliary price space $\mathcal{\tilde{P}}$ is non-empty, compact, convex, and closed under $f_{TF}(\cdot)$, and $F(\cdot)$ is defined as the convex hull of Cromme and Diener's $H_f$. Therefore, as shown in Cromme and Diener, Kakutani's fixed point theorem implies that $F(\cdot)$ has a fixed point $\tilde{p}^*\in \mathcal{\tilde{P}}$. Thus, a convex combination of the error terms $z(p')$ near $t(\tilde{p}^*)$ yields relaxed market clearing.
\end{proof}
The fixed point result in Proposition~\ref{T:fixedPoint} implies that for prices $p'$ in the neighborhood of the truncation of a fixed point $\tilde{p}^*$ of the price update function, induced market clearing error will be low. Therefore, we will likely be able to find price vectors that yield low market clearing error by iteratively applying $f_{TF}(\cdot)$ to its own output, until an apparent fixed point is approached.

\subsection{Proof of Proposition~\ref{T:priceUpdate}}

First, we show that auxiliary price space $\mathcal{\tilde{P}}$ is closed under $f_{TF}(\cdot)$. The range of $z(p)_j$ for $p \in \mathcal{P}$ is $[-1, (|N'| - 1)(1 + \epsilon)]$. Thus, assuming $\epsilon$ is small, $|z(p)_j / |N'|| < 1$. Therefore, it is easy to see that auxiliary price space $[-1, 1 + \bar{b}]^{|N'|}$ is closed under $f_{TF}(\cdot)$.

Next, we show that if $\tilde{p} \in \mathcal{\tilde{P}}$ is a fixed point of $f_{TF}(\cdot)$, then $p = t(\tilde{p})$ must induce relaxed market clearing.

Each agent price $p_j$ in $p = t(\tilde{p})$ must either equal $0$, equal $\bar{b}$, or be in $(0, \bar{b})$.

If $p_j = \bar{b}$, then no agent can afford to demand agent $j$, so $U_j = 1$ and $D_j = 0$, implying that $z(p)_j$ is negative, so $p_j$ cannot be a fixed point. If $p_j$ has been truncated by $t(\cdot)$, this would reduce it, and $z(\cdot)$ further reduces it.

If $p_j = 0$, then if there is $D_j > 0$ and $U_j = 0$, the net effect of $z(p)_j$ is to increase $p_j$, and if $p_j$ has been truncated, then that would have increased $p_j$ also; so with or without truncation, $p_j$ would not be a fixed point if it had positive $D$ but not $U$. If $p_j = 0$ and $D_j = 0$, then under-demand for $j$ is allowed because the price of $j$ is zero, so this does not violate relaxed market clearing. $D_j$ and $U_j$ cannot both be positive, because if there is incoming demand, there is no under-demand by definition.

If $p_j \in (0, \bar{b}$), and $p_j$ is part of a fixed point $p$, then $U_j = 0$ and $D_j = 0$, because by definition, they cannot both be positive at once. \qed


\subsection{Proof of Proposition~\ref{T:aceeiTFic}}

First, we sketch a proof of strategyproofness-in-the-large.

Any individual agent in the continuum economy, being zero-measure, has no influence on the approximate equilibrium price vector arrived at by update function $f_{TF}(\cdot)$, at any iteration of the A-CEEI-TF mechanism. Therefore, the only effect the agent can have on the outcome is that, if the agent is randomly selected to choose its favorite affordable team of free agents that leaves a feasible subproblem, the agent's reported preferences determine which team the agent is assigned. Thus, it is a dominant strategy for the agent to report its true preferences, so that in this case the agent will be assigned its most-preferred allowable team. \qed

Now we deal with approximate envy-freeness.

Consider a team formation problem where A-CEEI-TF yields exact market clearing with the same partition at each stage. If all agent budgets were equal, then the result would be envy-free. To see this, consider that each agent's team assigned by A-CEEI-TF in such cases is the same as the team induced by prices in the initial round of price search. In this initial round, every agent is able to afford the team of every other agent, because all agent budgets are equal, and all agents' demands are consistent (i.e., simultaneously satisfiable), because there is exact market clearing. Because the final partition returned is the same as that induced by the initial prices, no agent must envy any other in the final partition.

More generally, consider the same setting but without equal budgets. That is, consider a team formation problem with non-negative values and team size constraints, where A-CEEI-TF yields exact market clearing with the same partition induced at each stage. Recall that all agents' budgets are in the range $[1, \bar{b}]$, where $\bar{b} = 1 + (1 / |N|)$. Because $1 \leq \bar{k} \leq |N|$, this implies that $\bar{b} \leq 1 + (1 / \bar{k})$.

No agent $i$ can envy another agent $j$ in its own team, because $j$'s team (excluding the costs of $i$ and $j$) costs weakly less than the team of $i$ (excluding the cost of $i$), as all agent prices are non-negative, so $i$ could afford the team of other agents that $j$ receives (the team of agent $i$, excluding agents $i$ and $j$).

Assume for a contradiction that $i$ has envy for agent $j$ that is not bounded by a single teammate. Note that $i$ and $j$ must be assigned to distinct teams. Let $k'$ be the number of other agents in $j$'s team, $k' < \overline{k}$. Let $j_1$ be the first other agent in $j$'s team, and so on through $j_{k'}$.

Let $p$ be the exact market clearing prices that are found in the first stage of the A-CEEI-TF mechanism. Let $x$ be the partition that is induced by the mechanism, which we have assumed to be consistent across all stages of the price search. $x_i$ is then a vector in $\{0, 1\}^{|N|}$ that has a $1$ for each agent on the same team as agent $i$. For convenience in calculating the cost of the team for agent $i$, we let $x_{ii} = 0$ for all $i$, because an agent is not required to pay the cost of teaming with itself. Let $(x_i \setminus \{j\})$ be a vector equivalent to $x_i$, but with $x_{ij}$ set to $0$.

$p \cdot (x_j \setminus \{j_r\}) > b_i: \forall r \in \{1, \ldots, k'\}$. This is because if agent $i$'s envy for agent $j$ is not bounded by a single teammate, agent $i$ must not be able to afford any subset of $j$'s team created by removing one agent (and excluding agent $j$).

$(k' - 1) p \cdot x_j > k' b_i$. This comes from adding the L.H.S. and R.H.S. of the $k'$ inequalities above, where the subtracted teammates from $x_j$ cancel to produce one fewer complete $x_j$.

$b_j \geq p \cdot x_j$, because agent $j$ can afford its own team.

$(k' - 1) b_j > k' b_i$. Replacing $p \cdot x_j$ with $b_j$ in the inequality above.

$b_j / b_i > k' / (k' - 1)$.

Recall that $\bar{b} \leq 1 + (1 / \bar{k})$, and $\bar{k} \geq 1$, meaning that $\bar{b} \leq 1 + (1 / (\bar{k} - 1))$. This implies that $\bar{k} / (\bar{k} - 1) \geq \bar{b}$.

$\bar{k} > k'$, which implies $\bar{k} k' - \bar{k} < \bar{k} k' - k'$. So $\bar{k} (k' - 1) < k' (\bar{k} - 1)$. This means $\bar{k} / (\bar{k} - 1) < k' / (k' - 1)$.

$b_j / b_i > k' / (k' - 1) > \bar{k} / (\bar{k} - 1) > \bar{b}$.

$b_j / b_i > \bar{b}$. 

This contradicts the claim that all budgets are $\in [1, \bar{b}]$. Therefore, envy must be bounded by a single teammate in a team formation problem with non-negative values and team size constraints, where A-CEEI-TF yields exact market clearing with the same partition at each stage. \qed

\subsection{A-CEEI-TF Implementation}

Unlike the alternative mechanisms, A-CEEI-TF is far from straightforward to implement.
Below we describe an implementation which uses ideas from~\citeauthor{othman2010finding}~\shortcite{othman2010finding}.

To run the A-CEEI-TF mechanism, we used a modified form of the procedure described by Othman, Sandholm, and Budish, which alternates tabu search and mixed-integer program (MIP) solving until an approximate fixed point in price space is found \cite{othman2010finding}. The mechanism starts by randomly assigning approximately equal budgets to the agents  $\in (100, 100 + 100 / |N|)$, and setting prices for all agents to be equal to the minimum agent budget divided by the maximum number of agents that can be legally demanded, $(\overline{k} - 1)$. Based on the random order of players, the next unselected player in order is allowed to select its favorite affordable team of unselected players, of legal size and that leaves a feasible subproblem. Before a player can make this selection, the tabu search-MIP hybrid algorithm is run to assign the prices to remaining agents.

Each time the price search procedure is called, $20$ iterations of tabu search are run. For a given iteration, the current price vector is ``expanded'' to yield up to $5 |N| + 5$ new price vectors, each of which is evaluated in terms of its relaxed market clearing error. 

The new price vectors are generated as either unilateral deviations from the current price vector, or moves in the direction of the ``gradient,'' or $z(\cdot)$. The gradient neighbors of $p$ are generated by moving in the direction of $z(p)$, for each of the step sizes $(10.0, 5.0, 1.0, 0.5, 0.1)$, where the step size is just the $L_2$-norm of the change in $p$. For each agent $j$, we generate unilateral neighbors of $p$ by setting the agent's price to $0$ if $z(p)_j \leq 0$, otherwise by increasing the agent's price by each of $(1.0, 0.5, 0.1, 0.05, 0.001)$. Any neighbor prices outside of $[0, \bar{b}]$ are truncated to this range.

Each neighbor price vector of the previous price $p$ is evaluated for its market clearing error, using an MIP solver to determine which team each agent would demand at those prices. We used CPLEX v12.5, with Concert bindings for Java, to solve each MIP. Market clearing error for a price vector $p$ is defined as the $L_2$-norm of $z(p)$.

A ``step'' of the tabu search entails moving from the previous price vector to the price vector of lowest error among the neighbors, that is not in the tabu list (a list of the previously visited price vectors). After $20$ iterations, the visited price vector that produced the lowest (best) market clearing error is returned.

To find which team an agent $j$ will demand, given a set of legal team sizes that leave feasible subproblems, we first break the legal team sizes into groups of consecutive integers. For example, if $\underline{k} = 2, \overline{k} = 4$, and there are $4$ agents remaining including the ``self'' agent, it is not legal for the agent to demand $2$ other agents for a team of size $3$, even though $3 \in [2, 4]$, because this does not leave a feasible subproblem. So in this case we would first break up the problem into two distinct MIPs, one with $\underline{k} = \overline{k} = 2$, the other with $\underline{k} = \overline{k} = 4$. The agent's preferred team from the union of the MIP result teams would then be returned as the agent's demand.

Now we have a price vector $p$ over the $M$ other remaining agents and a consecutive sequence of integers that are legal team sizes, $\{\underline{k}, \underline{k} + 1, \ldots, \overline{k}\}$. The demand of agent $j$ with budget $b_j$ is:

\begin{align*}
\max_c : \sum_{i=1}^{|M|} c_i u_j(i) \\
\textrm{subject to:} \\
c \in \{0, 1\}^{|M|} \\
\sum_{i=1}^{|M|} c_i p_i \leq b_j \\
\underline{k} - 1 \leq \sum_{i=1}^{|M|} c_i \leq \overline{k} - 1
\end{align*}